\documentclass[a4paper,twoside]{article}

\usepackage{graphicx} % Required for inserting images
\usepackage{import}
\usepackage[dvipsnames]{xcolor}
\usepackage{xspace}
\usepackage{multirow}
\usepackage{booktabs}
\usepackage{amsthm}
\usepackage{amsfonts}
\usepackage{algpseudocode}
\usepackage{algorithm}
\usepackage{comment}

\usepackage{epsfig}
\usepackage{subcaption}
\usepackage{calc}
\usepackage{amssymb}
\usepackage{amstext}
\usepackage{amsmath}
\usepackage{multicol}
\usepackage{pslatex}
\usepackage{apalike}
\usepackage[bottom]{footmisc}
\usepackage{SCITEPRESS}

\newcommand{\ourmethod}{SONNI\xspace}
\newcommand{\newatk}{Silver Platter\xspace}

\newcommand{\merc}{server\xspace}
\newcommand{\client}{client\xspace}
\newcommand{\modelowner}{provider\xspace}

\newcommand{\Client}{Client\xspace}
\newcommand{\Modelowner}{Provider\xspace}

\newcommand{\pr}{\ensuremath{\mathcal{P}}\xspace}
\newcommand{\cl}{\ensuremath{\mathcal{C}}\xspace}
\newcommand{\se}{\ensuremath{\mathcal{S}}\xspace}

\newtheorem{theorem}{Theorem}

\begin{document}

\title{SONNI: Secure Oblivious Neural Network Inference}

%\author{\authorname{Withheld for Double-blind Review}}
\author{\authorname{Luke Sperling and Sandeep S. Kulkarni }
\affiliation{Department of Computer Science and Engineering \\
Michigan State University}
\email{\{sperli14, sandeep\}@msu.edu }
}
\keywords{Fully Homomorphic Encryption, MLaaS, Privacy.}

\bibliographystyle{apalike}% the mandatory bibstyle
\abstract{%\begin{abstract}
    In the standard privacy-preserving Machine learning as-a-service (MLaaS) model, the client encrypts data using homomorphic encryption and uploads it to a server for computation. The result is then sent back to the client for decryption. It has become more and more common for the computation to be outsourced to third-party servers. In this paper we identify a weakness in this protocol that enables a completely undetectable novel model-stealing attack that we call the  \newatk attack. This attack works even under multikey encryption that prevents a simple collusion attack to steal model parameters.
    We also propose a mitigation that protects privacy even in the presence of a malicious server and malicious client or model provider (majority dishonest). When compared to a state-of-the-art but small encrypted model with 32k parameters, we preserve privacy with a failure chance of $1.51 \times 10^{-28}$  while batching capability is reduced by $0.2\%$. %batching capability reduction?

    Our approach uses a novel results-checking protocol that ensures the computation was performed correctly without violating honest clients' data privacy. 
    Even with collusion between the client and the server, they are unable to steal model parameters.
    Additionally, the model provider cannot learn any client data if maliciously working with the server.
%\end{abstract}
}

%\begin{document}
\onecolumn \maketitle \normalsize \setcounter{footnote}{0} \vfill
%\maketitle
%\abstract{\input{01-abstract}}

\section{Introduction \label{sec:intro}}

%\lnote{Should I add a motivating example such as a hospital using a MLaaS skin cancer assessment?}

%first four paragraphs of conclusion here:

%Generating ML models is expensive (add text, what are examples)
Training ChatGPT3\cite{brown2020language} took $3.14 \times 10^{23}$ floating point operations or an estimated $\$4.6$M if replicated using a Tesla V100 cloud instance\footnote{https://lambdalabs.com/blog/demystifying-gpt-3}. QuillBot uses a variety of models for grammar checking, paraphrasing, etc, some of which contain billions of parameters and a took team of engineers to develop\footnote{https://quillbot.com/blog/compressing-large-language-generation-models-with-sequence-level-knowledge-distillation/}.

%In biometrics systems, templates can be stolen from the database

% One way to compensate is to use MLaaS
One way for a \modelowner to recuperate the cost of generating machine learning models is to provide machine learning-as-a-service (MLaaS) where the \modelowner permits a client to compare their data with the model. For example, a hospital could pay to have automated screenings for skin cancer. Moreover, millions of users pay for their queries to be answered by ChatGPT or to have their writings improved by QuillBot's premium offerings. This incentivizes the costly training process with financial gain.  

% Issue of privacy in MLaaS for client and its importance.

The use of MLaaS generates an issue of privacy for the \client, as the data it is using is private and the \client is concerned about possible leak of this private data to \modelowner. The issue of privacy is especially important to \client as the data is often one-of-a-kind and cannot be replaced/regenerated, such as in the case of medical or biometric data. For this reason, for several years, MLaaS systems have been designed with \client data privacy in mind. 

% This was ok if \modelonwer in the computational device

However, as MLaaS becomes more widespread and the models become more complex, the \modelowner lacks the computational power needed to provide MLaaS on its own. Especially in the context of availability of third-party servers to provide computational capability, it is in the interest of the \modelowner to outsource the MLaaS service to a \merc that will only provide the computational capability. Even in this context, the client privacy is crucial. Additionally, in this context, the privacy of the \modelowner is critical as well. Specifically, since the development of the model requires significant computational and financial investment, it is critical that the model remains a secret and not revealed to anyone including the \client and the \merc. 

Our work focuses on using oblivious neural networks with mulitikey encryption \cite{chen2019efficient}.
%Oblivious neural network inference using HE has been devised to combat this by way of multikey encryption, where 
Here, both the client and model provider each hold a secret key and both are needed for decryption. The result is computed by the third-party server and sent back to the client and model provider. The model provider applies their secret key to get a partial decryption and sends it back to the client who is able to fully decrypt and learn the result.

The protocol in \cite{chen2019efficient} provides privacy only for the honest-but-curious and collusion models. In other words, if the server, model provider or client is only going to try to infer from the data it has learned, the privacy property is satisfied. However, if the server is malicious and violates the protocol even slightly, it can leak model parameters without being detected.
%colludes with the client, then it can leak model parameters without being detected. 
%while secure in the honest-but-curious model of computation, lends itself to a serious vulnerability if a stronger threat model is considered. 
One way to do this is to encrypt the model parameters as the return value and send them to the client instead of sending the actual result of the model. In this case, the server will not be able to detect this as the data it sees is encrypted with the client's key. It follows that this attack would go completely undetected by the model provider 100\% of the time. 
%For instance, what if the result computed in the encrypted domain is not the output of the model applied on the input but rather the result just contains model parameters? This attack would go completely undetected by the model provider 100\% of the time. 

A simple solution would be to have the model provider check the result before returning it to the client but this is insufficient. The output of the model applied to the private data is, itself, private data as well. Additionally, it is possible in certain applications to recreate the original data from the output \cite{mai2018reconstruction}. Therefore, there is a strong need for a solution that ensures the result was computed correctly without revealing what the output is to the model provider.

In this paper, we identify a novel attack against oblivious neural network inference and propose a secure protocol that we call Secure Oblivious Neural Network Inference, or \ourmethod, to combat it and ensure security in the presence of a dishonest server and dishonest \client/\modelowner (majority dishonest).

%\lnote{old content below:}

%This protocol prevents the ML model provider from learning the sensitive data $x$ or the output of the function $f(x)$. Additionally, this protocol prevents the client from learning any model parameters $f$ outside of those that can be learned via black box attack \cite{}. There are many papers that focus on mitigating black box attacks \cite{}. These privacy guarantees only apply under the semi-honest and collusion threat models. While it is reasonable to assume that the hired server will not be byzantine, there is a specific attack that is reasonable that the server can perform in this situation with no chance of being detected. The server may deploy a client process with the goal of learning the model parameters. Instead of returning $enc(f(x))$, the server can simply return $enc(f)$. The client then decrypts the resultant ciphertext to receive $f$.

%In this paper, we tackle the problem of the server cooperating with the client in order to perform a novel model stealing attack. Additionally, we address the possibility that the server may instead be working with the ML model provider to learn $x$ or $f(x)$. While in the protocol described in \cite{chen2019efficient} no such attack is possible, protocols that seek to mitigate this attack may open up the possibility of this attack.

The contributions of the paper are as follows:

%\snote{review again.}

\begin{itemize}
    \item We introduce a novel attack, called the \newatk attack, to enable model stealing in the oblivious neural network inference setting.
    \item We propose a novel protocol for oblivious neural network inference that mitigates this novel attack and does not violate the client's data privacy.
    \item We evaluate our protocol in terms of the level of privacy it provides and the overhead of providing that privacy.
\end{itemize}

%\snote{Organization of the paper. }
In Section \ref{sec:background} we provide background of the encryption schemes employed in this work. In Section \ref{sec:atk} we detail the vulnerability of the existing outsourced neural network inference protocol and the novel attack we construct. We propose a solution and prove its security in Section \ref{sec:def}. We provide related work in Section \ref{sec:rel} and closing remarks in Section \ref{sec:conclusion}.

\begin{figure*}
    \centering
    \includegraphics[width=0.65\textwidth]{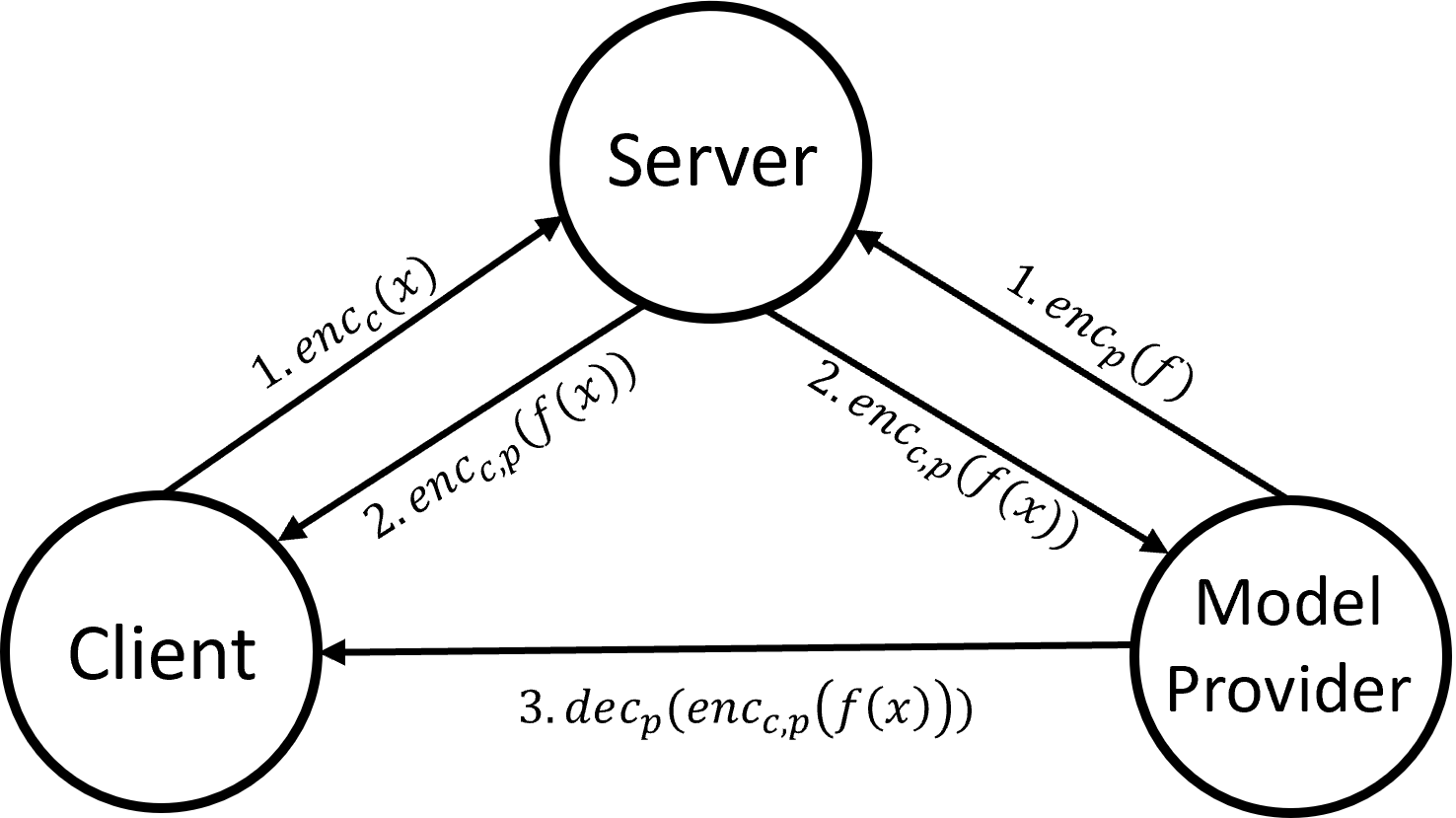}
    \caption{Standard Outsourced Oblivious Neural Network Inference Model}
    \label{fig:standard_model}
\end{figure*}
\section{Background \label{sec:background}}
\subsection{Fully Homomorphic Encryption}

Traditional encryption schemes require decryption before any data processing can occur. Homomorphic encryption (HE) allows arithmetic directly over ciphertexts without requiring access to the secret key needed for decryption, enabling data oblivious protocols.
Fully homomorphic encryption (FHE) schemes support an unlimited number of multiple types of operations. The Cheon-Kim-Kim-Song (CKKS) \cite{cheon2017homomorphic} cryptosystem supports additions as well as multiplications of ciphertexts. 
The message space of CKKS is vectors of complex numbers.

A small amount of error is injected into each ciphertext at the time of encryption to guarantee the security of the scheme. This error grows each arithmetic operation performed. Therefore, once a certain threshold of multiplications are performed the error grows too large to decrypt correctly. Although the bootstrap operation reduces this error to allow circuits of unlimited depth, for the purposes of this work no such operations are needed. The following operations are supported:
\begin{itemize}
    \item Encode: Given a complex vector, return an encoded polynomial.
    \item Decode: Given a polynomial, return the encoded complex vector.
    \item Encrypt: Given a plaintext polynomial and public key, return a ciphertext.
    \item Decrypt: Given a ciphertext and the corresponding secret key, return the underlying polynomial.
    \item Ciphertext Addition: Given two ciphertexts, or a ciphertext and a plaintext, return a ciphertext containing an approximate element-wise addition of the underlying vectors.
    \item Ciphertext Multiplication: Given two ciphertexts, or a ciphertext and a plaintext, return a ciphertext containing an approximate element-wise multiplication of the underlying vectors.
    \item Relinearization: Performed after every multiplication to prevent exponential growth of ciphertext size. This function is assumed to be executed after every homomorphic multiplication in this work for simplicity.
    \item Ciphertext Rotation: Given a ciphertext, an integer, and an optionally-generated set of rotation keys, rotate the underlying vector a specified number of times. For instance, $Rotate(Enc([1,2,3,4]),1,k_r)$ returns $Enc([2,3,4,1])$.
\end{itemize}

\subsection{Multikey Encryption}
Chen et al. \cite{chen2019efficient} extended existing FHE schemes to operate on multiple sets of secret key, public key pairs. Consider the case where there are two key pairs owned by two parties (such as the \Client and \Modelowner). Each party encrypts a value using their respective public keys. These ciphertexts, denoted by $enc_c(X)$ and $enc_p(Y)$ can be multiplied together homomorphically, resulting in $enc_{c,p}(XY)$. Note that this notation means the ciphertext is now multikey encrypted by the keys of both the \Client and the \Modelowner and that $enc_{k_1,k_2}(V) = enc_{k_2,k_1}(V)$. The following additional operations are supported:
\begin{itemize}
    \item $PartialDec(ciphertext, secret\_key)$ - given a multikey encrypted ciphertext and one of the secret keys corresponding to a public key encrypting the ciphertext, return a partial decryption of the ciphertext. This partial decryption does not yield any additional information.
    \item $Combine(partial\_decryptions)$ - given an array of partial decryptions, return the underlying plaintext value. This combine function only works if there is a partial decryption corresponding to each public key used in the multikey encryption.
\end{itemize}
%\lnote{define terminology $enc_{c,p}$ and describe that $enc_{c,p}=enc_{p,c}$}
\section{Proposed Model Stealing Attack \label{sec:atk}}

In this section, we demonstrate that the existing approach for using multiparty encryption to provide privacy in outsourced MLaaS systems suffers from a serious violation of privacy for the \modelowner if the server behaves maliciously and tries to leak parameters to the client. We call this attack the \newatk attack, as the adversary is  delivered the exact model parameters with little effort. Additionally, this attack will go unnoticed by the \modelowner.

%\subsection{System and Adversary Model \lnote{Perhaps we drop the subsection header if there will only be a single subsection}}

The outsourced oblivious neural network inference system based on multi-key FHE \cite{mukherjee2016two} consists of three parties, the \client, the model \modelowner and the \merc (to which the computation is outsourced). 

The \client holds some private data $x$ for use as input to a parameterized function $f$. The model provider, referred to henceforth as the \modelowner, holds a private function $f$ that is comprised of private parameters. The \merc holds no private data and is instead responsible for computing $f(x)$ given encryptions of $f$ and $x$.

In this model, there are two privacy requirements: the client wants to ensure that the values of $x$ and $f(x)$ are not released to anyone while the \modelowner wants to ensure that none learns the (parameters of) function $f$. However, existing client server models typically focus on the privacy requirement of the client and ignore the privacy requirements of the provider or server \cite{boulemtafes2020review,yang2023review}.

%Specifically, \cite{?}
%Client-server models for privacy-preserving frameworks typically 
%assume the semi-honest threat model \cite{} or the collusion threat model \cite{}. 
These threat models are insufficient, considering how rampant model stealing attacks have become \cite{oliynyk2023know}. Hence, we assume the byzantine threat model where any of the parties may be byzantine. If the \merc is byzantine, it may be cooperating with either the \client (to steal the model $f$) or with the \modelowner (to steal $x$ or $f(x)$). 

In cases where the \merc is colluding with the \client, it will try to steal some/all model parameters in place of the value for $f(x)$. 
In cases where the \merc is cooperating with the \modelowner, it is in the best interest of both of them to return $f(x)$ to the \client as failure to do so would raise suspicion on the \merc and/or \modelowner and it will cause reputational/financial  damage to one or both of them.
However, the \merc and \modelowner may exchange additional messages that allow them to learn $f(x)$ or $x$

%\lnote{Adversary capabilities:} in our model we presume the adversary can send arbitrary messages in place of, or in addition to messages expected by the protocol.

%\snote{subscript to enc}

The communication between \client, \modelowner and \merc is as shown in Figure \ref{fig:standard_model} where we utilize the oblivious neural network inference protocol from \cite{chen2019efficient}. The \client encrypts their data using fully homomorphic encryption (FHE) and the \modelowner encrypts their function parameters using a different public key. These encryptions (along with the corresponding public keys) are sent to the \merc. The \merc then computes $enc_{c,p}(f(x))$ which is returned to both the \client and the \modelowner. Partial decryptions are obtained with each party's secret key and are combined by the \client to obtain the full decryption of $f(x)$.

The approach in \cite{chen2019efficient} preserves the privacy of the client as the \modelowner and \merc only sees the encrypted version of the data. It also preserves the privacy of \modelowner if the \merc is honest-but-curious, as the \merc only sees the data in encrypted format. However, if \merc is dishonest, it can collude with \client to reveal the model parameters. We discuss such attacks next. 

In step 2 of Figure \ref{fig:standard_model}, a malicious server sends an encryption of the parameters $enc_{c,p}(f)$ instead of the intended result $enc_{c,p}(f(x))$. $enc_{c,p}(f)$ can be trivially computed from $enc_p(f)$ which is obtained by the server during step 1. This ciphertext is fully decrypted only by the client during step 3, revealing private model parameters. This attacks appears identical to normal operation to the provider. In cases where there are multiple ciphertexts with model parameters, the attack may be repeated to obtain them all.

\begin{comment}
To describe this attack, we consider a simple model that consists of a single fully-connected linear layer. The output of the model on an $i_d$ dimensional vector $x$ is $f(x) = Ax$ for a parameter matrix $A$ of shape $o_d \times i_d$. 

Observe that the \merc has access to the parameters of $f$ that are encrypted by \modelowner's public key, i.e., $enc_p(f)$. In other words, it has access to $enc_p(A)$, where $A$ is the matrix used by the model. 
This \merc is supposed to send $enc_{c,p}(f(x))$ to the \modelowner. However, suppose that the dishonest \merc sends the response as $enc_{c,p}(A')$, where $A'$ contains a row of matrix $A$. The \modelowner will decrypt this message to obtain $enc_{c}(A')$. Since this is encrypted message, the \modelowner has no ability to know that the \merc is dishonest. Instead, it will send it to the client. The client will decrypt it to obtain $A'$, which corresponds to a row of $A$.
\end{comment}

The reason we call this the \newatk attack is that this can be repeated ad infinitum without the \modelowner learning anything about the fact that the model is being compromised. This paper focuses on addressing this attack.

\section{Attack Mitigation via Oblivious Results Checking \label{sec:def}}
In this section we propose a protocol that mitigates the \newatk attack and we prove its security. First, in Section \ref{sec:problemstatement} we introduce the problem formally. Next, in Section \ref{sec:mitigationapproach} we describe our approach at a high level before going into details in Section \ref{sec:protocoldetails}. We then describe the special steps needed to prepare the ciphertext at the beginning of the computation as well as the work needed to check that the result was computed correctly. Lastly, in Section \ref{sec:security} we analyze our protocol and prove its security.
%\lnote{create an outline for the entire section and put it here}

\begin{figure*}
    \centering
    \includegraphics[width=0.85\textwidth]{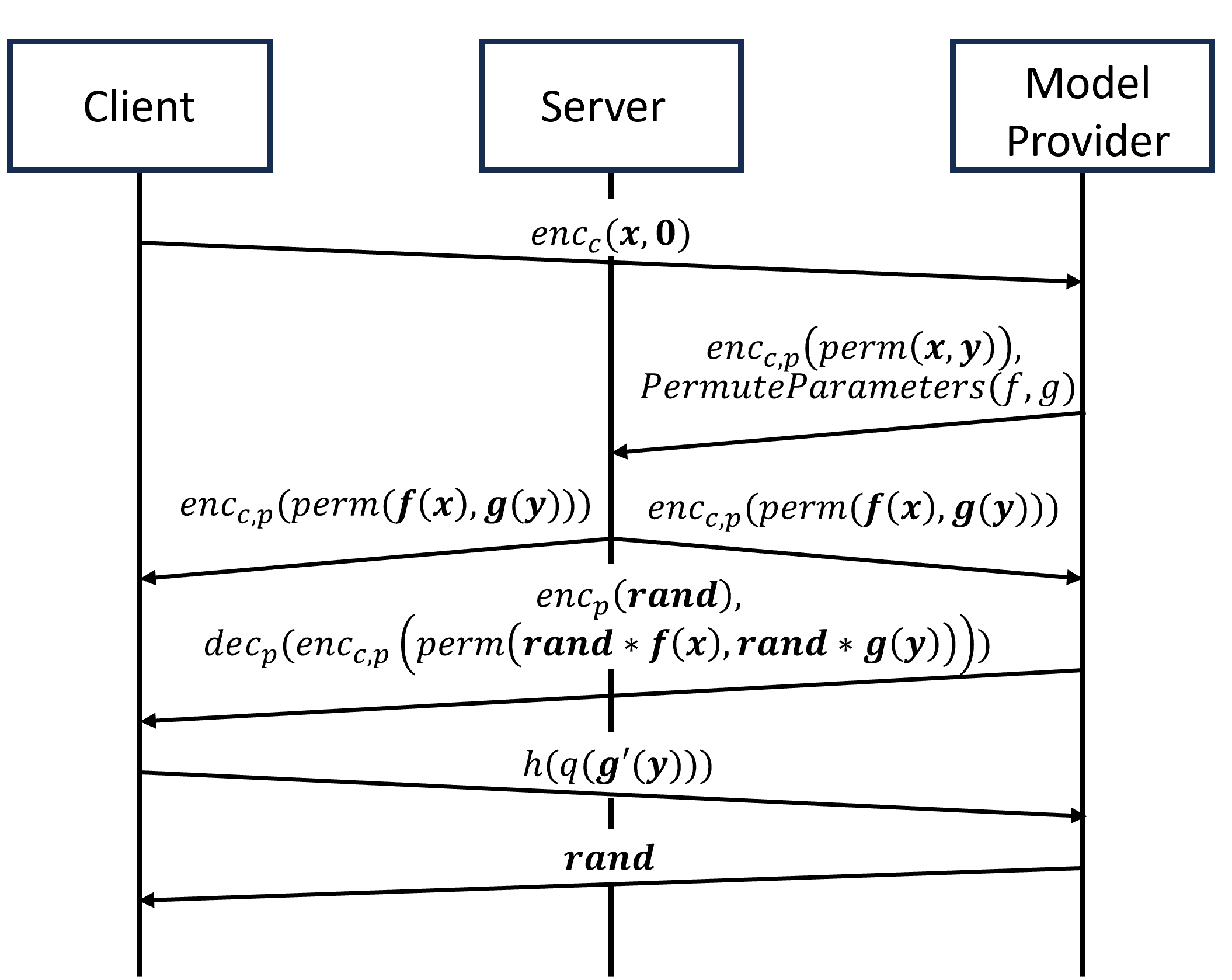}
    \caption{\ourmethod sequence diagram.}
    \label{fig:sequence}
\end{figure*}

\subsection{Problem Statement \label{sec:problemstatement}}
%Given \client secret data $x$ and \modelowner secret function ($f$) parameters $w_i$, a third party server, referred to as the \merc, works with the 

In this section, we define the problem statement for dealing with \newatk attack. Here, a party (\client, \merc or \modelowner) may be honest or dishonest. An honest party follows the prescribed protocol, while a dishonest party could be byzantine and behave arbitrarily. We assume that the dishonest party intends to remain hidden. Therefore, it will send messages of the required type (such as CKKS ciphertexts). However, the data inside may be incorrect. 

Recall that in the multiparty communication, the client starts with its own data, $x$, and desires to obtain a value $f(x)$ where $f$ is a function known to \modelowner. The actual computation of $f(x)$ is done by a \merc.

In the multiparty communication considered in this paper, we follow a \textit{secure with abort} approach where security is guaranteed at all times, i.e., the \modelowner or \merc never learns the client data, $x$ or the computed value $f(x)$ but if dishonest behavior is detected, the computation may be aborted. This remains true in all cases, i.e., even if the \modelowner and \merc send additional messages between them privacy will not be violated. We also want to ensure that the model parameters are never revealed to the client even if the \merc and \client collude and send additional messages between them. However, if everyone behaves honestly then the client should receive the value of $f(x)$. Thus, the privacy requirements are as follows:

\begin{description}
    \item [Liveness] Upon providing $enc_c(x)$ to \modelowner, the \client should eventually learn $f(x)$ if the \modelowner and \merc behave honestly.
    \item [\Modelowner privacy] The \client should not learn the model parameters even if the \merc is colluding with the \client.
    \item [\Client privacy] The \modelowner should not learn $x$ or $f(x)$ even if the \merc is colluding with the \modelowner. 
    \item [Oblivious Inference] The \merc should not learn input $x$, model parameters, or result $f(x)$.
    \item [Secure with Abort] If any party violates the protocol then the computation is aborted and the \client does not learn anything.
\end{description}

% If the \merc is working with the \modelowner, the goal is to learn private data $x$ or 
% function output $f(x)$. The computation is considered complete with the \client learns $f(x)$. 
% No process should learn any information about $x$ or $w$, and the \client and \modelowner should not learn $f(x)$. 

As stated, we follow the \textit{secure with abort} model for handling collusion between the \merc and \client, i.e., the \modelowner will abort the computation if it detects that the \merc or \client is trying to steal model parameters. In this case, the \client does not receive $f(x)$. It also does not learn any model parameters. 

In our protocol, the \client does not need to abort as there is no possibility of the \modelowner and \merc learning $x$ or $f(x)$.

%and likely blacklist the \merc in the event of any suspicious activity or if they suspect that the result was not computed correctly. 

\subsection{Mitigation Approach \label{sec:mitigationapproach}}

%\lnote{introduce the idea of working on vectors which gives way to ciphertext}
%\lnote{introduce the idea that we are going to dedicate some slots of the ciphertext to results checking at the opportunity cost of not using them in the main computation}

In this section, we propose our approach for providing privacy to \modelowner as well as the \client. Specifically, the approach in \cite{chen2019efficient} already provides privacy for the \client data. We ensure that the corresponding privacy is preserved while permitting the \modelowner to have control over the data being transmitted to the \client so that it can verify that the \merc is not trying to leak the \modelowner's private data (namely the model parameters) to the \client. 

Our approach is based on the observation that machine learning algorithms operate on vectors as opposed to scalars, making FHE well-suited to the task. Ciphertexts in FHE encode and encrypt entire vectors of real values. We take advantage of this fact for our mitigation strategy. We dedicate some slots of the ciphertext (and therefore some indices of the underlying vector) to verify the computation is performed correctly. These slots will not contain $f(x)$ values but rather randomly-generated values for use in a results checking step.

The basic idea of the protocol is to have the \merc compute $f(x)$ in most slots of the ciphertext and $g(y)$ in other slots, where $g$ is a function that is of the same form as $f$ (e.g., if $f$ is a linear/quadratic function then $g$ is also a linear/quadratic function with different parameters.). Here, $g$ is a random function generated at the start of the protocol by the \modelowner and $y$ is a random secret input vector also generated by the \modelowner. Upon receiving the answer from the \merc, the \modelowner will get the decryption of $g(y)$ by requesting it from the client. And, if the value of $g(y)$ matches then the \modelowner will return the result to the \client. Since $g$ and $y$ are only known to the \modelowner, if $g(y)$ is computed correctly then it provides a confidence to the \modelowner that the \merc is likely to have performed honestly. If $g(y)$ is not computed correctly, then the \modelowner would realize that either \merc or \client is being dishonest and can refuse to provide the partial decryption of the encrypted value of $f(x)$ necessary to learn $f(x)$.

We observe that the use of the extra slots for $g(y)$ reduces the number of slots available to the client. Also, the approach for providing privacy is probabilistic, i.e., there is a possibility that a dishonest \merc could use some of the slots used by $f$ to leak information without being caught. This probability will be low if the number of slots used for $g$ are high. However, in this case, the overhead will be higher as the number of slots available for $f$ would be lower. We discuss this in Section \ref{sec:security} and show that the probability of a successful model stealing attack against our system is negligible.

%This increases the size of the vector needed to perform the computation. We show in Section \ref{sec:security} that the number of slots needed for good security is much smaller than the size of the input vector.

\subsection{Protocol Details \label{sec:protocoldetails}}

\ourmethod is built on top of the existing oblivious neural network inference protocol with added steps. The protocol is detailed in Algorithm \ref{alg:overview}.

First in step \ref{step:submit}, the \client generates a ciphertext containing the vector $x$ in the first $d$ slots of the ciphertext and sends it to the \modelowner. 
The \modelowner inserts the $y$ values in slots $d+1..d+m$, where $m$ slots are used for $y$.
Subsequently, \modelowner permutes this vector in step \ref{step:permute1}. 
It also permutes function $f$ and $g$ in a similar fashion. As an illustration if the $f$ was $x[1]+ 2x[2]$ and $(x[1],x[2])$ is permuted to $(x[2], x[1])$ then $f$ would be changed to $2x[2]+x[1]$. 

Additionally, to \textit{permute the computation}, the \modelowner can also change $f$ to be  $2x[2]+x[1] + 0*y[1]$, where $y[1]$ is one of the values used for the computation of $g(y)$. Likewise, the computation of $g$ can use some of the values in the $x$ vector with coefficient $0$. It could use any arithmetic operation that uses $y$ (respectively, $x$) while computing the value of $f(x)$ (respectively $g(y)$) in such a way that the final answer will be independent of the $y$ (respectively $x$) value. This will prevent the \merc from performing data-flow techniques to determine the slots used for $f$ and $g$. 

Since the \merc is not aware of the placement of $x$ and $y$ in the ciphertext and it uses the permuted input, it computes permuted $f(x)$ and $g(y)$ in step 6. 
This value is sent to both the \client and the \modelowner.
%The \merc simultaneously computes $f(x)$ and $g(y)$ in the same ciphertext by applying the parameter ciphertexts in step ?.
%Due to this, both $f$ and $g$ will be of the same \lnote{form} but will have unique parameters. The results ciphertext is passed to both the \client and \modelowner.

%I think the below isn't true so I commented it out 
%In step ?, the \modelowner reveals which indices corresponding to the values of $g(y)$ to the \client. These values are masked out via ciphertext multiplication.

In step \ref{step:maskclient}, the \client uses the mask provided by the \modelowner to compute $enc_{c,p}(g'(y))$, where $g'(y)$ denotes $g(y)$ masked by a random vector known to the \modelowner (similarly $f'(x)$ denotes $f(x)$ masked by that same random vector). This value and the partial decryption provided by \modelowner is used to compute $g'(y)$.
%The extracted values are jointly decrypted by the \modelowner and the \client such that only the \client learns $g(y)$. 
If $g'(y)$ is computed correctly, then the model owner is convinced that the entire vector was computed correctly.

Since the computation of $g'(y)$ by the client may have a small error due to the use of FHE whereas the value computed by the \modelowner is the exact value of $g'(y)$, we use $g'_c(y)$ to denote the value computed by the \client and $g'_p(y)$ to be value computed by the \modelowner.
Each value is quantized such that $q(g'_c(y))$ will be equal to $q(g'_p(y))$ 
%(cf. Section \ref{?} for details.)

%the value of $g(y)$ is quantized to compute $q(g(y))$ (e.g., if $g(y)$ was supposed to be an integer, then the client will identify the closest integer to the value computed by it.) 
Then, in steps \ref{step:resultcheck1}-\ref{step:resultcheck4}, the \client and \modelowner perform a privacy-preserving comparison protocol to ensure $q(g'_c(y))$ was computed correctly. After this passes, the \client and \modelowner jointly decrypt the result ciphertext containing $f(x)$ and only the \client learns $f(x)$.
One of the concerns in this step is a malicious \modelowner that tries to compromise the privacy of the \client. Specifically, in this step, if the \modelowner specifies the incorrect indices such that the indices correspond to $f'(x)$ instead of $g'(y)$, then the \modelowner could learn the value of $f(x)$. To prevent this, the proof that the \client computes the value of $g'(y)$ correctly is achieved via a zero-knowledge approach discussed later in this section.

\begin{algorithm*}
\footnotesize
    \caption{Secure Oblivious Neural Network Inference }\label{alg:overview}
    \begin{algorithmic}[1]
        
        \State \cl sends to \pr: $enc_c(x[1],x[2],...x[d], 0,...0)$ \label{step:submit}
        
        \item[]

        \State \pr computes $enc_{c,p}(x[1],x[2],...x[d], y[1], y[2], y[m])$ \label{step:append}
        \State \pr computes $enc_{c,p}(perm(x,y))$ \label{step:permute1}
        \State \pr computes $PermuteParameters(f, g)$ using the same permutation. \label{step:permute2}
        \State \pr sends to \se: \mbox{$enc_{c,p}(perm(x,y))$, $PermutaParameters(f, g)$} \label{step:sendpermute}

        \item[]
        
        \State \se computes and sends to \pr and \cl: $enc_{c,p}(perm(f(x),g(y)))$
        %\State \pr computes $enc_{c,p}(perm(0,g(y))) = mult(enc_{c,p}(perm(f(x),g(y))), mask)$
        \State \pr computes $rand = [-1,1]^{d+m}$
        \State \pr computes $enc_{c,p}(perm(f'(x),g'(y))) = mult(enc_{c,p}(perm(f(x),g(y))), rand)$
        \State \pr sends to \cl: $dec_{p}(enc_{c,p}(perm(f'(x),g'(y)))), enc_p(rand)$
        \State \cl computes $g'(y)$ \label{step:maskclient}  %$ = dec_{c}( \ \ dec_{p}(enc_{c,p}(perm(0,g(y))))  \ \ )$ 

        \item[]

        \State \cl sends to \pr: $hash_1 = h(q(g'(y)))$ \label{step:resultcheck1}
        \State \pr computes $g'(y)$ \label{step:resultcheck2}
        \State \pr computes $hash_2 = h(q(g'(y)))$ \label{step:resultcheck3}
        \State If $hash_1 \neq hash_2$, \pr aborts \label{step:resultcheck4}

        \item[]
        
        %\State \pr sends to \cl: $dec_{p}(enc_{c,p}(perm(f(x),g(y))))$
        \State \pr sends to \cl: $rand$
        \State \cl computes $f(x)$

        %\State The \client sends $enc_c(x)$ to the \modelowner
        %\State \pr computes $enc_p(f[1], f[2], ... f[d], g[1],g[2],...g[m])$
        %\State The \modelowner generates $m$ random $y$ values and $m$ random $z$ values. The $z$ values comprise the parameters of $g$
        %\State The \modelowner appends $y$ values, then shuffles the ciphertext
        %\State The \modelowner sends $enc_c(perm(x,y))$ and all parameter ciphertexts $enc_p(w_i, z)$ (which comprise $f$ and $g$) to the \merc
        %\State The \merc computes $enc_{c,p}(perm(f(x),g(y)))$ and sends it to both the \client and the \modelowner
        %\State The \modelowner multiplies binary mask $mask \in \{0,1\}^{d+m}$ to $enc_{c,p}(perm(f(x),g(y)))$ to obtain $enc_{c,p}(perm(0,g(y)))$
        %\State The \modelowner sends $mask$ and partial decryption $dec_{p}(enc_{c,p}(perm(0,g(y))))$ to the \client
        %\State The \client computes then partially decrypts $enc_{c,p}(perm(0,g(y)))$ and combines with  $dec_{p}(enc(perm(0,g(y))))$ to obtain $g(y)$
        %\State The \client sends $hash_1 = h(q(g(y)))$ for a shared quantization scheme $q$ and hash function $h$
        %\State The \modelowner computes $hash_2 = h(q(g(y)))$ using the $g$ and $y$ generated
        %\State If $hash_1 \neq hash_2$ the \modelowner aborts
        %\State The \modelowner sends partial decryption $dec_{p}(enc(perm(f(x),g(y))))$ to the \client
        %\State The \client partially decrypts $enc_{c,p}(perm(f(x),g(y)))$ and combines with  $dec_{p}(enc(perm(f(x),g(y))))$ to obtain $f(x)$
    \end{algorithmic}
\end{algorithm*}

The above protocol has two important stages, ciphertext preparation (lines \ref{step:append}-\ref{step:sendpermute}) and result checking (\ref{step:resultcheck1}-\ref{step:resultcheck4}). Next, we provide additional details of these steps. Ciphertext preparation refers to the manipulation done by the \modelowner before it sends the data to \merc to compute $f(x)$ and $g(y)$ whereas the result checking refers the the computation between the \client and \modelowner so that it can verify that $g(y)$ was computed correctly by the server.

\subsubsection{Ciphertext Preparation} 
The security of our proposed protocol relies on the \client and \merc not knowing the indices of the $y$ values. To achieve this, the \client appends $m$ zeroes to the end of their data vector $x$ before encryption. This ciphertext is sent to the \modelowner. Then, through a series of masking multiplications, rotations, and additions, the ciphertext is shuffled as shown in Algorithm \ref{alg:shuffle}. We note that we do not require arbitrary ciphertext shuffling for the security of our protocol. Our security rests on the \merc not being able to guess which slots correspond to $x$ values and which slots correspond to $y$ values. The vectors of function parameters are shuffled in the same manner prior to encryption. The \modelowner finally adds random $y$ values to the zero indices in the ciphertext.
%\footnote{\snote{Is this paragraph repetitive?}}

%\subsubsection{Ciphertext Preparation} \snote{This title does not seem right here. }\lnote{I named it this due to the fact that before the computation, manipulation needs to be done to the input ciphertext to prepare it for the protocol}
%The security of our proposed protocol relies on the \client and \merc not knowing the indices of the $y$ values. To achieve this, the \client appends $m$ zeroes to the end of their data vector $x$ before encryption. This ciphertext is sent to the \modelowner. Then, through a series of masking multiplications, rotations, and additions, the ciphertext is shuffled as shown in Algorithm \ref{alg:shuffle}. We note that we do not require arbitrary ciphertext shuffling for the security of our protocol. Our security rests on the \merc not being able to guess which slots correspond to $x$ values and which slots correspond to $y$ values. The vectors of function parameters are shuffled in the same manner prior to encryption. The \modelowner finally adds random $y$ values to the zero indices in the ciphertext.

\begin{algorithm}
    \caption{Ciphertext shuffle}\label{alg:shuffle}
    \begin{algorithmic}
        %$\forall i \in d$: 
        %\State $x_i \gets mult(x, e_i)$ \lnote{I hear this $e_i$ is a notation for zero vector with a single 1 element.}
        %\EndForall

        %\For{$i=0$ to $m$}
        %\State $index \gets random(0,d+i-1)$
        %\State $maskedvalue \gets mult(ct, e_{index})$ \lnote{$e_i$ means all zeroes with a single one}
        %\State $maskedvalue \gets rotate(maskedvalue,d+i-index)$
        %\State $ct \gets mult(ct,\mathbf{1} - e_{index})$ \lnote{the operand is all ones with a single zero}
        %\State $ct \gets add(ct,maskedvalue)$
        
        %\EndFor

        \State $indices \gets \emptyset$
        \For{$i=0$ to $m$}
            \State $index \gets random(0,d+m)$
            \While{$index \in indices$} 
                \State $index \gets random(0,d+m)$
            \EndWhile
            \State $indices.add(index)$
        \EndFor
        \State $indices.sort()$
        \State $i \gets 0$
        \For{$index \in indices$}
            \State $maskedvalue \gets mult(ct, e_{index})$ \Comment{}{$e_i$ denotes a zero vector with a 1 at index $i$}
            \State $maskedvalue \gets rotate(maskedvalue,d+i-index)$
            \State $ct \gets mult(ct,\mathbf{1} - e_{index})$ %\Comment{the operand contains all ones with a single zero at index $i$}
            \State $ct \gets add(ct,maskedvalue)$
            \State $i \gets i+1$
        \EndFor
        
        %\State $k \gets random-int(0,m)$
        %\State $$
    \end{algorithmic}
\end{algorithm}

\subsubsection{Result Checking}
%\lnote{replace the zeroing out f(x) with multiplying each slot by a random value then revealing the random numbers used on successful result check.}

If the \modelowner does not verify that the ciphertext they are decrypting does not contain model parameters, the \newatk attack will succeed. However, the \modelowner should not learn $f(x)$ so it is not permissible to have the \modelowner decrypt the result, examine it, and return it to the \client. Indeed, if this were allowed, the protocol would be open to an attack where the \modelowner works with the \merc to learn $x$ in a similar manner. Our idea is to give the \client access only to the computed $g(y)$. If the \client is able to correctly tell the \modelowner the values of $g(y)$, then the \modelowner will help the \client decrypt the final result and gain access to $f(x)$.

%%%%%%%%%%%%%%%%%%%%%%%%%%%%%%%
%\lnote{old version begin}

%The \modelowner generates a binary mask vector $mask \in \{0,1\}^{d+m}$ that contains in each index a 1 if the corresponding slot of the results ciphertext contains a value of $g(y)$ and $0$ if the corresponding slot contains a value of $f(x)$. This vector is sent to the \client. Both the \modelowner and the \client compute $maskedresult = mult(mask,result)$ along with a partial decryption of $maskedresult$. The \modelowner sends their partial decryption to the \client, who combines the partial decryptions to recover $g(y)$. 

%The \modelowner computes $g(y)$, as they generated both $g$ and $y$ in plaintext prior to the computation. Now, the \modelowner needs to be convinced that they have the same value of $g(y)$ as the \client without revealing these values to each other. If the \client were to send $g(y)$ to the \modelowner, the \modelowner could take advantage of this by storing $f(x)$ or $x$ in lieu of $g(y)$ (or an obfuscated version thereof). Similarly, the \modelowner cannot simply send $g(y)$ to the \client because the \client could simply lie and state that the values are the same.

%%%%%%%%%%%%%%%%%%%%%%%%%%%%%%%
%\lnote{new version begin}

The \modelowner generates a random vector $rand \in \mathbb{R}^{d+m}$ that contains in each index a random nonzero real value. This vector is sent to the \client. Both the \modelowner and the \client compute:

\begin{flalign} 
enc_{c,p}&(perm(f'(x),g'(y))) \nonumber\\ 
&= mult(enc_{c,p}(perm(f(x),g(y))), rand) \nonumber
\end{flalign} 

%$$enc_{c,p}(perm(f'(x),g'(y)))= mult(enc_{c,p}(perm(f(x),g(y))), rand)$$

They also compute a partial decryption of this resultant ciphertext. The \modelowner sends their partial decryption to the \client that combines the partial decryptions to recover $g'(y) = g(y) \times rand$ (and additionally $f'(y) = f(y) \times rand$). 

The \modelowner computes $g'(y)$, as they know $g$, $y$, and $rand$ due to having generated each in plaintext at the start of the computation. Now, the \modelowner needs to be convinced that they have the same value of $g'(y)$ as the \client without revealing these values to each other. If the \client were to send $g'(y)$ to the \modelowner, the \modelowner could take advantage of this by storing $f(x)$ or $x$ in lieu of $g'(y)$ (or an obfuscated version thereof). Similarly, the \modelowner cannot simply send $g'(y)$ to the \client because the \client could simply lie and state that the values are the same.

%%%%%%%%%%%%%%%%%%%%%%%%%%%%%%%
%\lnote{new version end}

A secure two-party computation protocol is needed to prove to the \modelowner that the values of $g'(y)$ are the same without revealing those values to the other. This protocol can work on plaintext values because $g'(y)$ exists in plaintext to both parties at this point in the protocol (as decryption has already occurred). Both parties compute $h(g'(y))$ for some low-collision hash function $h$. The \client sends their result to the \modelowner and the \modelowner verifies that the hashes are identical.

Due to the approximate nature of arithmetic under FHE, the value of $g'(y)$ decrypted by the \client is likely to be slightly different than the value computed in plaintext by the \modelowner. To mitigate this, we use a simple uniform quantization scheme with intervals large enough that the noise from homomorphic computations will not affect the performance. Therefore, the \client and \modelowner compute $h(q(g'(y)))$ instead of $h(g'(y))$ for a quantization scheme $q$. 
%The interval width is explored both theoretically in Section \ref{sec:security} and empirically in Section \ref{sec:experiments}. \lnote{remove this sentence?}

%Once the \modelowner is convinced that $g(y)$ was computed correctly, they participate in decrypting the results ciphertext and share the permutation ordering with the \client to recover $f(x)$.

Once the \modelowner is convinced that $g'(y)$ was computed correctly, they share the random vector $rand$ with the \client to for use in recovering $f(x)$.

\subsection{Security Analysis \label{sec:security}}
%\lnote{Need to consider quantization interval width?}

\begin{figure*}
    \centering
    \includegraphics[width=0.49\textwidth]{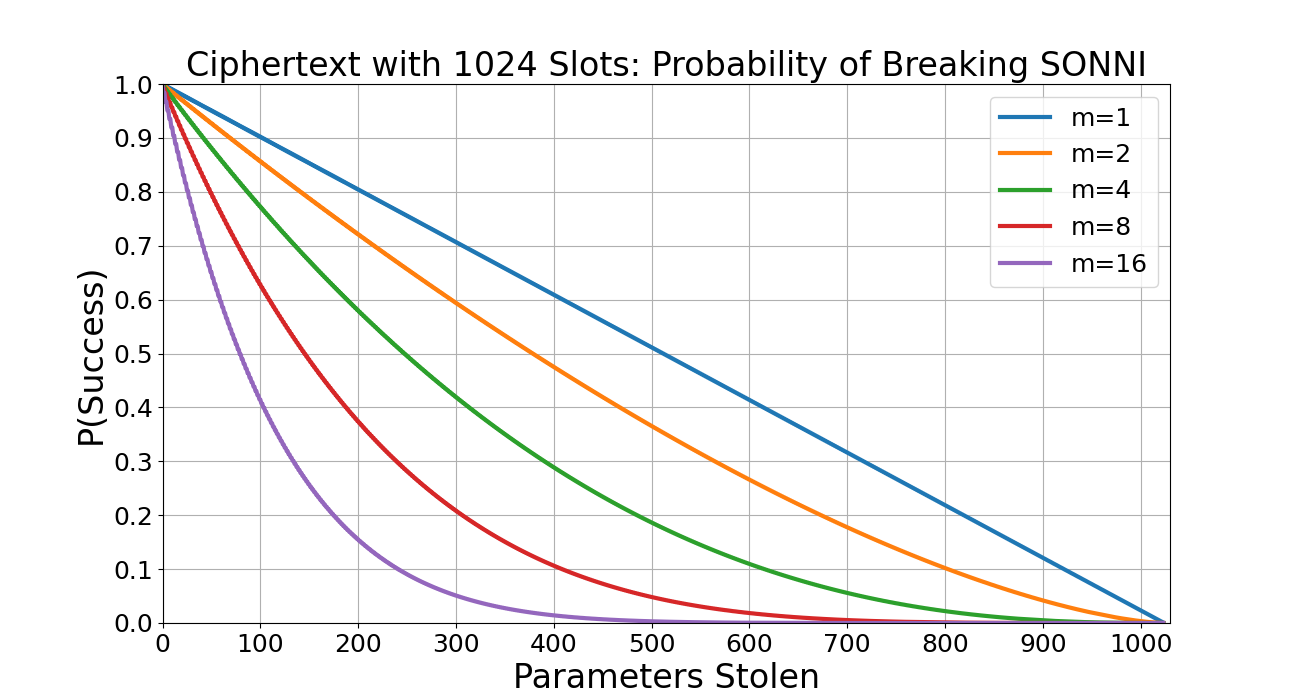}
    \includegraphics[width=0.49\textwidth]{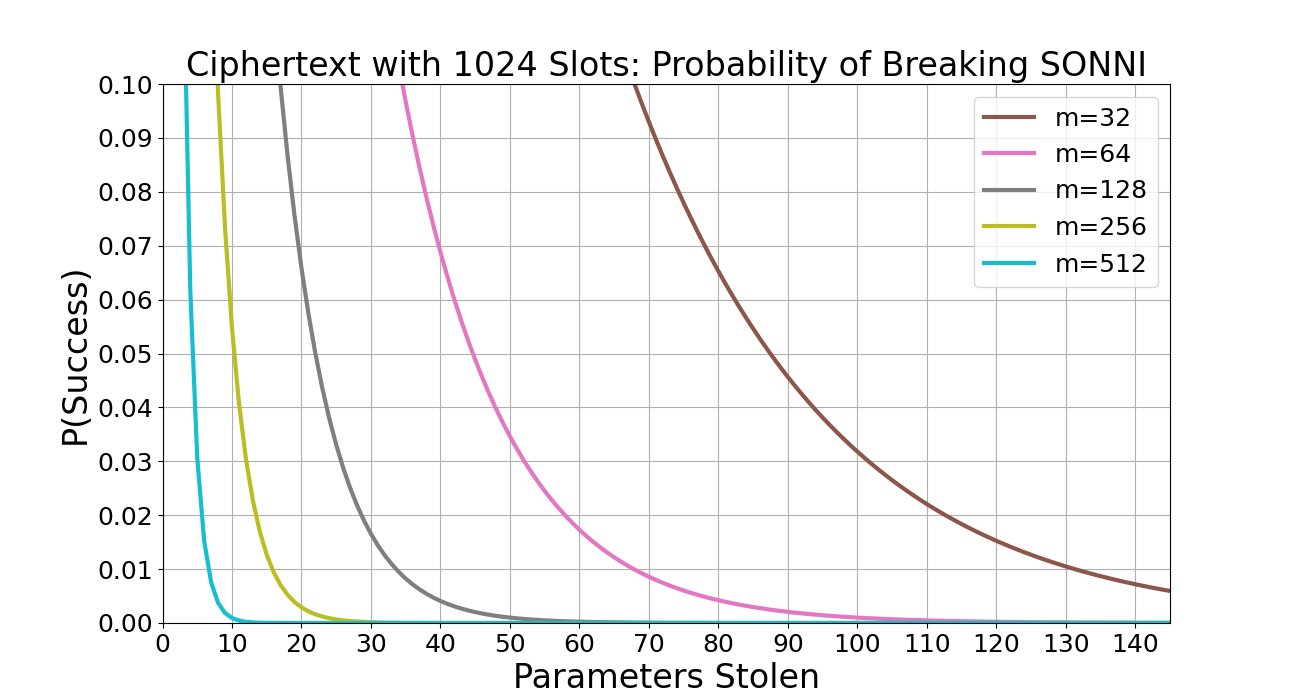}
    \caption{Probability of performing a \newatk attack without being detected based on the number of ciphertext slots dedicated to checked values.}
    \label{fig:smuggling_success}
\end{figure*}

In this subsection, we evaluate the theoretical probabilities of attacks successfully working against the proposed system. We refer to the probability of breaking the homomorphic encryption scheme as $\epsilon_1$ and the probability of breaking the incorporated hash function as $\epsilon_2$. These values are negligible, given proper selection of encryption security parameters and hash function.

\begin{theorem} \label{theorem:silverplatter}
Let $d$ be the dimensionality of the \client's secret data $x$ and $m$ be the number of ciphertext slots dedicated to containing $y$ values. The probability of successfully performing the \newatk attack against \ourmethod to steal $k$ scalar model parameters is upper bounded by ${(\frac{d}{d+m})}^k + \epsilon_1 + \epsilon_2$.
\end{theorem}

\begin{proof}
The probability of successfully performing the attack depends on the \merc passing the result checking step of the protocol. The \modelowner performs the ciphertext shuffle locally and, due to ciphertext indistinguishability, the \merc is unable to differentiate any permutations. Therefore, the \merc must guess which slots of the ciphertext will not be checked during the result checking step. To hide a single scalar value in the result ciphertext without being found out results in one of the $d$ slots dedicated to the results being selected instead of one of the $m$ slots, or a probability of ${\frac{d}{d+m}}$. Stealing a second value in this ciphertext is therefore a probability of ${\frac{d-1}{d+m-1}}$ which is less than the original probability, so the optimal strategy is to independently steal one parameter per result ciphertext for a total probability of ${(\frac{d}{d+m})}^k$ for $k$ model parameters. The only other ways to successfully perform the attack are to break the hash function or the homomorphic encryption scheme, thus the theorem is proved. We note that no modifications may be made to the ciphertext by the \client, lest the decryption fail (from the \client and \modelowner attempting to partially decrypt different ciphertexts).
\end{proof}

\begin{theorem}
The probability of the \modelowner working with the \merc to successfully steal any amount of \client data $x$ or $f(x)$ is upper bounded by 
$\epsilon_1 + \epsilon_2$.
\end{theorem}

\begin{proof}
%Note that the \modelowner and the \merc never view any decryptions during the protocol. The only new information the \modelowner \lnote{sees} during the protocol is the ciphertext containing $x$, the ciphertext containing $f(x)$ and $g(y)$, and the hashes of $g(y)$. By working with the \merc, the \client can be induced to send $h(x)$ instead of $h(g(y))$. This can be done by simply inserting $x$ values in the indices where $g(y)$ are intended to be stored. In this case, in order to learn $x$, the \modelowner must either break the hash $h(x)$ or decrypt the ciphertext containing $x$, which is computationally equivalent to guessing the secret key of the \client.
%\lnote{Maybe only the following sentence is necessary:}
The \modelowner and \merc together only learn encryptions and hashes during the protocol. If the \modelowner is able to break the hash, a new attack is enabled. The \merc can manipulate the results ciphertext such that the \client returns $h(q(f'(x))))$ in place of $h(q(g'(x))))$ to the \modelowner who can then learn $q(f'(x))$ via breaking the hash. Therefore, being able to break either a hash or an encryption is the only way to learn new information about $x$ or $f(x)$, proving the theorem.
\end{proof}

\begin{table}[ht]
\small
    \centering

        \caption{Probability of successful \newatk attack against \ourmethod. Vector size of $1024$ is used. \label{tab:probs}}
    
    \begin{tabular}{|c|c|c|}
    
    \hline
    m & Parameters Stolen & $P(Success)$\\

    \hline
    \multirow{4}{*}{4} & 10 & 0.952 \\
     & 128 &  $0.513$\\
     & 256 &  $0.237$\\
     & 512 &  $0.0311$\\
    \hline
    \multirow{4}{*}{32} & 10 & $0.720$ \\
     & 128 &  $0.011$\\
     & 256 &  $6.34 \times 10^{-5}$\\
     & 512 &  $7.07 \times 10^{-11}$\\
    \hline
    \multirow{4}{*}{512} & 10 &  $9.25 \times 10^{-4}$\\
     & 128 &  $2.88 \times 10^{-43}$\\
     & 256 &  $1.02 \times 10^{-96}$\\
     & 512 &  $1.23 \times 10^{-307}$\\
    \hline
    
    \end{tabular}

    %\label{tab:accuracy}
    %\vspace*{-10mm}
\end{table}

In Figure \ref{fig:smuggling_success}, we analyze Theorem \ref{theorem:silverplatter} quantitatively. We assume in this analysis that the attack will be performed in one round, rather than stealing one parameter per transaction (as that becomes incredibly costly from a financial perspective when the number of parameters reaches thousands or millions). In FHE applications it is a common practice to use 128-1024 size vectors \cite{jindal2020secure}. We consider the case where the ciphertext has 1024 slots, some of which ($d$) are allocated to the \client and some ($m$) are allocated to the \modelowner.
When $m=1$, i.e., only one slot is allocated to the \modelowner for verification of $g(y)$, the probability of a successful \newatk is quite high.
However, even with a few slots for the \modelowner, the probability of a successful \newatk decreases sharply. Table \ref{tab:probs} demonstrates this. With $m=32$, i.e., $32$ slots allocated to \modelowner, the probability of stealing 10 parameters is 0.$720$, the probability of stealing 128 parameters is $0.011$, and so on. The probability of stealing $512$ parameters is $7.07 \times 10^{-11}$. We note that modern machine learning models often include thousands or millions of parameters. The simple projection matrix model featured in \cite{sperling2022heft} includes $32768$ parameters. Even if the adversary maximizes their chance of successfully performing the attack by stealing a single model parameter at a time over $32768$ transactions and if only $2$ ciphertext slots dedicated to verification of $g(y)$, the probability of stealing all these parameters is $1.51 \times 10^{-28}$. Stealing just half of the parameters has a chance of $1.23 \times 10^{-14}$ of success. Batching capability is reduced by only $0.2\%$ while privacy is maintained with only a negligible chance of being violated.

In the case where even a small number of model parameters leaking is undesirable, using an m value of $512$ ensures that \newatk attacks are extremely unlikely to be successful. Even stealing 10 in this circumstance has only a $9.25 \times 10^{-4}$ chance of succeeding.

We note that a failed attempt at model stealing in this circumstance is more dire than a failed attempt at breaking an encryption scheme, for instance.
%We consider that this probability is acceptable, as it requires a \merc to be dishonest and the probability of a dishonest \merc is low.
It is anticipated that the \merc and \modelowner will have a financial relationship which would be compromised if the \merc is found to be dishonest. Also, the \merc is likely to be a company that provides a server farm. That means that risking this relationship for dishonest behavior that is very likely to be caught is not worthwhile. This may risk legal or financial consequences.

Also, our protocol preserves \client privacy. This is especially important for a \client as its data, $x$ or $f(x)$, is often impossible to replace, as the data corresponds to medical information, biometric information, etc. By contrast, revealing a very small number of parameters does not significantly affect the \modelowner. 
As discussed in Section \ref{sec:atk}, the \newatk is undetectable for the \modelowner if one uses the existing approach \cite{chen2019efficient}. With the level of security provided to both the \client and \modelowner, we can provide privacy guarantees to both of them with a reasonable overhead.

\section{Related Work \label{sec:rel}}

%\begin{itemize}
%    \item model stealing attacks (how our work is diff)
%    \item collusion attacks
%    \item FHE?
%    \item private agreement / verification?
%\end{itemize}

Oblivious neural network inference aims to evaluate machine learning models without learning anything about the data and has seen much attention in recent years \cite{mann2023towards}. XONN achieves this with a formulation similar to garbled circuits \cite{riazi2019xonn}. There is a privacy concern related to outsourcing these computations. Namely, that the model will be learned by the server. FHE has been used for inference that not only protects sensitive client data, but also prevents learning the \modelowner's private data (i.e. the model) \cite{chen2019efficient}. As we discuss thoroughly in this work, this approach is open to the undetectable \newatk Attack. Our work protects not only client data, as is expected in the classical oblivious neural network inference formulation, but also that of the \modelowner.

Model stealing attacks aim to learn models from black-box access. Broadly speaking, this either attempts to steal exact model values or general model behavior. Our work examines the former scenario. This falls into three categories: stealing training hyperparameters \cite{wang2018stealing}, stealing model parameters \cite{lowd2005adversarial,reith2019efficiently}, and stealing model architecture \cite{oh2019towards,yan2020cache,hu2019neural}, often using side-channel attacks. 

Gentry's seminal work constructing the first FHE scheme \cite{gentry2009fully} paved the way for more privacy-preserving applications on the cloud. Concerns over data privacy led to the development of multi-key encryption featuring multiple pairs of public and private keys. The first use of multi-key FHE was proposed by López et al. \cite{lopez2012fly}. Work followed suit, improving speed as well as supporting more modern schemes \cite{chen2019efficient,chen2019multi,ananth2020multi}. These schemes are developed for secure cloud computation. Specific optimizations have been proposed for federated learning \cite{ma2022privacy}.

Threshold HE schemes requiring a subset of key holders collaborating to decrypt ciphertexts have been devised as early as 2001 \cite{cramer2001multiparty,damgaard2003universally}. Threshold FHE schemes have also seen popularity over the years both shortly after the advent of FHE \cite{myers2011threshold} and in recent years using modern FHE schemes \cite{jain2017threshold,sugizaki2023threshold,chowdhury2022efficient,jain2017threshold}. As we demonstrate in this paper, these techniques alone are not sufficient to ensure data privacy in the outsourced server setting, as these methods do not verify that the proper value is being computed prior to decryption.

Collusion attacks are often not considered in the context of FHE systems, as they tend to use the client-server model and only account for the privacy of the client, even in the biometrics setting where privacy is highly desired \cite{engelsma2022hers,drozdowski2019application,sperling2022heft}. Our work shows that assuming the honest-but-curious threat model for the participants in the MLaaS setting opens the way for undetectable and 100\% successful attacks.
%Mult-key FHE has been developed for modern FHE schemes such as CKKS \cite{chen2019efficient} and has been optimized for the federated learning setting \cite{}.  Threshold decryption schemes are superfluous for this work as this work requires two data owners to jointly decrypt a ciphertext (threshold decryption in the two-party case reduces to multi-key encryption). %Thus, multi-key FHE is sufficient. However, multi-key FHE by itself is not enough to guarantee privacy if the server cooperates with the malicious party, as this paper demonstrates.

\section{Conclusion \label{sec:conclusion}}
%MLaaS systems have been designed with client data privacy in mind for multiple years. The advent of proprietary model providers outsourcing computation to third party servers has revealed a new concern of data privacy - that of the model provider. Machine learning models are time consuming and expensive to produce, not only taking many hours to train but also taking many resources just to gather and prepare extensive datasets. 

%Model stealing attacks jeopardize the ability of companies to profit off of their proprietary models. Necessary protocols are needed, such as those presented in this paper, to protect the privacy of the model provider. Indeed, using FHE it is now simple to protect just the privacy of the client or the model oproviderwner. In this paper, we demonstrated the challenges of protecting both the client and model provider's privacy simultaneously and proposed protocols that effectively protect all party's privacy even when the server cooperates with the malicious party.

%\lnote{we are always concerned about data privacy for client, but now there's a lot of concerns over model stealing and model owner privacy. generating the model is expensive}

In this paper, we demonstrated that the current approach to oblivious neural network inference \cite{chen2019efficient} suffers from a simple attack, namely the \newatk attack. It permits the \merc to reveal all model parameters to the \client in a way that is completely undetected by the \modelowner. 

We also presented a protocol to deal with the \newatk attack. The protocol relied on the observation that FHE based techniques utilize vector-based computation. We used a few slots in this vector to compute a function $g(y)$ that could be verified by the \modelowner so that it can gain confidence that the \merc is not trying to steal the model parameters with the help of the client. Our protocol provided a tradeoff between the overhead (the number of slots used for verification) and the probability of a successful attack. We demonstrated that even with a small overhead in terms of the number of slots used for verification, the probability that a \merc can steal several parameters remains very small. Modern machine learning models often involve billions of parameters. We showed that the probability of stealing even 1\% of the parameters is very small.

%The advent of proprietary model owners outsourcing computation to third party servers has revealed a new concern of data privacy - that of the model owner. Machine learning models are time consuming and expensive to produce, not only taking many hours to train but also taking many resources just to gather and prepare extensive datasets. 

Future work directions include finding a modified shuffle method that allows the evaluation of functions that incorporate rotations to work on shuffled ciphertexts. Some matrix-vector multiplication methods require rotation, for instance, which is not supported by this work \cite{halevi2014algorithms}. Extensions of this work could also include devising a protocol that does not increase the multiplicative depth of the circuit, leading to a reduced overhead. Our protocol is secure against majority dishonest attacks, but it is not clear if it is the \merc or the \client who is dishonest when the results checking step fails. Future work can tackle the problem of what to do logistically when the check fails. If the \merc is malicious it should indeed be blacklisted, but this should not open the way to a dishonest \client attempting to get an honest \merc blacklisted. We leave that as an open problem.

% We considered what if the computatin is offloaded

% Generates a new type of attack

\bibliography{main}

\end{document}